\documentclass[aps,prl,twocolumn,superscriptaddress]{revtex4-2}
\usepackage{bm}
\usepackage{mathrsfs}
\usepackage{amsmath}
\usepackage{amssymb}
\usepackage{amsfonts}
\usepackage{amsthm}
\usepackage{graphicx}
\usepackage{color}
\usepackage{dcolumn}
\usepackage[T1]{fontenc}
\usepackage{multirow}
\DeclareMathOperator{\Tr}{Tr}

\newcommand{\eff}{\mathsf{E}}
\newtheorem{theorem}{Theorem}
\newtheorem{proposition}{Proposition}
\newtheorem{lemma}{Lemma}

\begin{document}

\title{Forbidden Subspaces in Quantum State Smoothing}

\author{Chon-Fai Kam}
\email{dubussygauss@gmail.com}
\affiliation{Dipartimento di Fisica e Chimica ``Emilio Segr\`e'', Universit\`a degli Studi di Palermo, Via Archirafi 36, I-90123 Palermo, Italy}
\affiliation{DSIMB, Inserm, BIGR U1134, Universit\'e Paris Cit\'e and Universit\'e de La R\'eunion, 75015 Paris, France}

\author{Kai-Wen Wong}
\affiliation{Life Actuarial Department, Taiwan Insurance Institute, 6th floor, No.~3, Nan-Hai Road, Taipei 100, Taiwan (R.O.C.)}

\date{\today}

\begin{abstract}
A system between a preparation and a post-selection has had no agreed state since 1964. With positivity as the criterion, a post-selection admits an interval of orderings around the symmetric one exactly when the subspace it forbids is spanned by eigenvectors of a full-rank filtered state. Otherwise it certifies contextuality, testable on a qubit. The averaged filtered state carries the entanglement spectrum of the record, so the smallest forbidden subspaces a symmetry allows are even-dimensional in the Haldane class and odd in the trivial one. That parity is the record's topological class.
\end{abstract}

\maketitle

A quantum system between a preparation and a post-selection has had no agreed state since Aharonov, Bergmann and Lebowitz asked for one \cite{abl1964}. The two records fix every probability in between \cite{aharonov1991,ritchie1991,dressel2014}, yet the operator that fixes them is no density operator. Its real part can go negative, the anomalous weak value \cite{aav1988}, which certifies contextuality \cite{pusey2014,kunjwal2019,spekkens2008} and is a resource \cite{arvidsson2024,debievre2021,langrenez2024,thio2025,tan2024,zhang2026}. Quantum smoothing is where that pair arrives uninvited \cite{tsang2009,tsang2022,guevara2015,chantasri2021,liu2025}. Classically the two records combine into $\beta_t\pi_t$, a filtered density times a backward likelihood \cite{rauch1965,mayne1966,fraser1969,bresler1986,briers2010}, and the quantum problem inherits the halves alone. The record before $t$ gives the filtered state $\rho_t$ and the record after $t$ the retrofiltered effect $\eff_t$, whose pairing $\Tr[\eff_t\,\mathcal{I}_m(\rho_t)]$ gives the probability of an outcome $m$ at $t$ \cite{gammelmark2013,tan2015}. A state at $t$ is another matter. The one defined by optimal estimation over the unobserved environment \cite{laverick2021,laverick2021cost} fails to be classical even where filtering and retrofiltering succeed \cite{laverick2023}.

The obstruction is one of ordering. Quantum mechanically $\beta_t\pi_t$ becomes the split orderings $S_s=\rho^{s}\eff\rho^{1-s}$ \cite{leetsutsui2017,leetsutsui2018,cahill1969}. Its ends $S_0=\eff\rho$ and $S_1=\rho\eff$ are the Kirkwood--Dirac operator and its adjoint \cite{kirkwood1933,dirac1945}. Only a Hermitian member is a state, so take
\begin{equation}
\Phi_s(\eff,\rho)=\tfrac12\big(S_s+S_s^{\dagger}\big),\qquad s\in[0,1],
\label{eq:phis}
\end{equation}
the Heinz mean of left and right multiplication by $\rho$ \cite{heinz1951,bhatia1997}, whose ends coincide in the Margenau--Hill construction $\Phi_0=\Phi_1=\tfrac12(\eff\rho+\rho\eff)$ \cite{margenau1961}. Every member has trace $\Tr[\eff\rho]$ and reduces to $\eff\rho$ when the two commute, and none is positive throughout \cite{kochen1967,fine1982,lostaglio2023,horsman2017}.  States over time pick the Jordan product at the Margenau--Hill end \cite{leifer2013,fullwood2022,lie2024,parzygnat2023bayes,parzygnat2023axioms}, quantum information picks the Kirkwood--Dirac operator \cite{hofmann2014,umekawa2023}, and neither has post-selection in view. Positivity picks neither end but the symmetric split. Taking positivity of the image as the criterion of classicality, we ask what neither has asked. Which post-selections leave the smoothed state classical?

Monitored systems are now a laboratory for phases of matter \cite{li2018,skinner2019,fisher2023,lavasani2021,sang2021,morral2023,xiao2026}, for a reason lying in the monitoring. A system interacting in turn with a train of modes leaves the record in a many-body entangled pure state whose bond is the system, and every matrix product state arises this way \cite{schon2005,fannes1992,perezgarcia2007}. Keeping a record writes in time the entanglement a chain carries in space. While no outcome is read the global state stays pure, so the cut at $t$ has a Schmidt decomposition whose spectrum \cite{lihaldane2008} is that of the object the past record supplies, and the symmetry-protected classification of matrix product states \cite{haldane1983,aklt1987,pollmann2010,pollmann2012,chen2011,chen2013,schuch2011} is read from that same spectrum. The answer to our question is topological. A post-selection forbidding a subspace is an effect with a null space, and for full-rank $\rho_t$ it admits a positive image on an open interval of orderings about $s=\tfrac12$ exactly when that subspace is spanned by eigenvectors of $\rho_t$ (Theorem~\ref{thm:invariant}). Under a symmetry the eigenspaces of the record-averaged $\rho_t$ are multiplets, whose dimensions are all even in the Haldane class and all odd in the trivial class \cite{pollmann2010}. The smallest symmetric post-selections compatible with a classical smoothed state are then even-dimensional in one class and odd in the other (Fig.~\ref{fig:concept}), and where none is the pair certifies contextuality. Each step is imposed by the one before it.

\begin{figure}[t]
\includegraphics[width=\columnwidth]{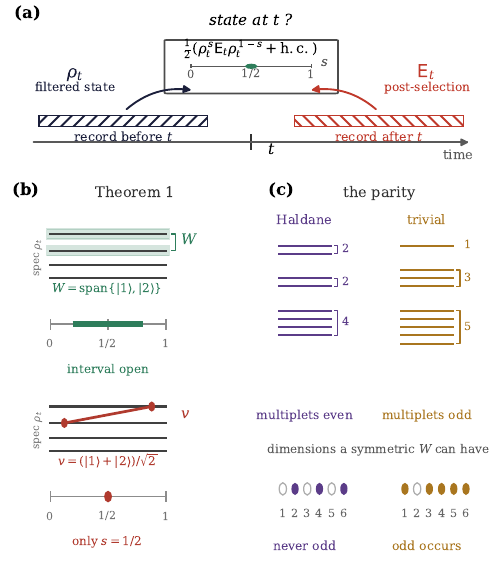}
\caption{(a) Quantum smoothing. The record before $t$ gives $\rho_t$ and the record after $t$ a post-selection $\eff_t$. Their pairing has a family of images indexed by $s$, only $s=\tfrac12$ being a state for every arrangement. (b) Theorem~\ref{thm:invariant}, for full-rank $\rho_t$. A forbidden subspace spanned by eigenvectors leaves an open interval admissible. Here $|1\rangle$ and $|2\rangle$ carry distinct eigenvalues, so their superposition is no eigenvector and leaves only $s=\tfrac12$. Widths are schematic. (c) Parity. A symmetric $W$ is a sum of whole multiplets, so its available dimensions are the sums drawn and filled here. They are even in the Haldane class and include odd ones in the trivial one.}
\label{fig:concept}
\end{figure}

\textit{What positivity selects.} Among the members of Eq.~\eqref{eq:phis} exactly one is a state for every $\rho$ and $\eff$, namely $\Phi_{1/2}=\rho^{1/2}\eff\rho^{1/2}$, a congruence of a positive operator. Positivity is therefore a question of how far from $s=\tfrac12$ one may move before it is lost, and the natural variable is the departure $\varepsilon=s-\tfrac12$. Factoring $\rho^{1/2}$ out of both sides of $\rho^{\frac12+\varepsilon}\eff\rho^{\frac12-\varepsilon}$ leaves $\rho^{\varepsilon}\eff\rho^{-\varepsilon}$ in the middle, and with $\mathcal{K}=[\log\rho,\,\cdot\,]$ this is $e^{\varepsilon\mathcal{K}}(\eff)$. The modular flow of $\rho$ is $\sigma^{\rho}_{u}(X)=\rho^{iu}X\rho^{-iu}=e^{iu\mathcal{K}}(X)$ \cite{takesaki1970,takesaki1972,connes1973}. Here that flow is applied to the effect and continued to imaginary time $u=-i\varepsilon$. Its adjoint is $e^{-\varepsilon\mathcal{K}}(\eff)$. The symmetrization in Eq.~\eqref{eq:phis} averages the two signs and gives the operator form of $H_{1/2+\varepsilon}(a,b)=\sqrt{ab}\cosh[\varepsilon\ln(a/b)]$,
\begin{equation}
\Phi_{\frac12+\varepsilon}(\eff,\rho)=\rho^{1/2}\cosh(\varepsilon\mathcal{K})(\eff)\,\rho^{1/2},
\label{eq:reduction}
\end{equation}
with $\cosh(\varepsilon\mathcal{K})(\eff)=\eff+\tfrac{\varepsilon^{2}}{2}[\log\rho,[\log\rho,\eff]]+O(\varepsilon^{4})$. Congruence by the invertible $\rho^{1/2}$ preserves positivity, so $\Phi_s$ is a state iff $\cosh(\varepsilon\mathcal{K})(\eff)\ge0$, and the identity answers the question it was built for. Positivity is lost at second order in the departure and never at first, since the first-order term $[\log\rho,\eff]$ is anti-Hermitian and the symmetrization removes it. The loss is symmetric, $\Phi_s=\Phi_{1-s}$, because $\cosh$ is even. How fast it is lost is set by the smallest eigenvalue of $\eff$ against the width of $\log\rho$, and by nothing about the dimension (End Matter~\ref{app:methods}). Multiplicativity singles it out again (End Matter~\ref{app:mult}).

\textit{A forbidden subspace.} The width is set by how far $\eff$ stands from singular, so the extreme case is a singular effect. That is a post-selection with a forbidden subspace, whose outcomes cannot precede the later record and which is the configuration of anomalous weak values \cite{aav1988,dressel2014,pusey2014}. There the interval can close entirely.

\begin{lemma}[A forbidden subspace closes the interval unless it is invariant]
\label{lem:rankdef}
Let $\rho>0$, $0\le\eff\le I$, and $\eff v=0$ for some unit vector $v$. Then
\begin{equation}
\big\langle v\,\big|\cosh(\varepsilon\mathcal{K})(\eff)\big|\,v\big\rangle
=-\varepsilon^{2}\,\big\langle (\log\rho)v\,\big|\,\eff\,\big|\,(\log\rho)v\big\rangle+O(\varepsilon^{4}),
\label{eq:nullleak}
\end{equation}
and the coefficient of $\varepsilon^{2}$ is strictly negative unless $(\log\rho)v\in\ker\eff$. If $(\log\rho)v\notin\ker\eff$, then $\Phi_s(\eff,\rho)$ is not positive for any $s\neq\tfrac12$ in some punctured neighbourhood of $\tfrac12$.
\end{lemma}

\begin{proof}
Write $L=\log\rho$, which is Hermitian because $\rho>0$. Set $X_{\varepsilon}=\cosh(\varepsilon\mathcal{K})(\eff)=\sum_{n\ge0}\varepsilon^{2n}\mathcal{K}^{2n}(\eff)/(2n)!$, a series of even powers converging for every $\varepsilon$ in finite dimension. The first two terms give Eq.~\eqref{eq:nullleak}. The constant term $\langle v|\eff|v\rangle$ vanishes because $\eff v=0$. The next is $\tfrac12\langle v|[L,[L,\eff]]|v\rangle$. Expanding $[L,[L,\eff]]=L^{2}\eff-2L\eff L+\eff L^{2}$ leaves only the middle term, since $\langle v|L^{2}\eff|v\rangle=\langle L^{2}v|\eff v\rangle=0$ and likewise for $\langle v|\eff L^{2}|v\rangle$. Hermiticity of $L$ gives $\langle v|L\eff L|v\rangle=\langle Lv|\eff|Lv\rangle$. The remainder is $O(\varepsilon^{4})$ because the omitted powers of $\varepsilon$ are even and at least fourth.

For the sign, positivity of $\eff$ gives $\langle w|\eff|w\rangle=\|\eff^{1/2}w\|^{2}$ for every $w$, an expression that is zero when $w\in\ker\eff$ and strictly positive otherwise. The coefficient of $\varepsilon^{2}$ in Eq.~\eqref{eq:nullleak} is therefore nonpositive, vanishing exactly when $Lv\in\ker\eff$ and strictly negative otherwise. It dominates the remainder for small enough $\varepsilon$, so $\langle v|X_{\varepsilon}|v\rangle<0$ for all $0<|\varepsilon|<\varepsilon_{0}$. Set $w=\rho^{-1/2}v$, which exists because $\rho>0$, so that Eq.~\eqref{eq:reduction} gives $\langle w|\Phi_{\frac12+\varepsilon}|w\rangle=\langle v|X_{\varepsilon}|v\rangle<0$ and $\Phi_s$ is not positive at those $s$.
\end{proof}

Numerically the failure extends to all of $[0,\tfrac12)$ in every one of $300$ random instances \cite{supp}, hence by $\Phi_s=\Phi_{1-s}$ to every $s\ne\tfrac12$. A forbidden subspace drawn at random leaves only the symmetric split. When it does not, the lemma gives a geometric condition. Imposing that on every null vector at once turns it into an exact statement about the subspace, which is the theorem.

\begin{theorem}[The interval survives a forbidden subspace iff the subspace is invariant]
\label{thm:invariant}
Let $\rho>0$ and $0\le\eff\le I$ with $\ker\eff\neq\{0\}$. Then $\Phi_s(\eff,\rho)\ge0$ for all $s$ in some open interval containing $\tfrac12$ if and only if $\ker\eff$ is invariant under $\rho$, that is, spanned by eigenvectors of $\rho$.
\end{theorem}

\begin{proof}
Since $\rho>0$ is Hermitian, a subspace is $\rho$-invariant iff it is spanned by eigenvectors of $\rho$. Such a subspace is invariant under $\log\rho$ as well, since $\log\rho$ acts on each eigenspace as a scalar. Conversely a $\log\rho$-invariant subspace is $\rho$-invariant, $\rho$ being a polynomial in $\log\rho$ in finite dimension, so the two invariances are the same condition on $\ker\eff$.

Suppose $\ker\eff$ is $\rho$-invariant. Since $\eff$ annihilates it and both operators are Hermitian, both preserve $\ker\eff$ and $(\ker\eff)^{\perp}$, and both are block diagonal in that splitting. Every power $\rho^{s}$ is then block diagonal too. Writing $\eff'$ and $\rho'$ for the restrictions to $(\ker\eff)^{\perp}$ gives $\Phi_s(\eff,\rho)=0\oplus\Phi_s(\eff',\rho')$, the first block vanishing because $\eff$ does there. On that subspace $\rho'>0$ and $\eff'>0$, the latter because $\eff\ge0$ and $\ker\eff$ has been removed. So $\Phi_{1/2}(\eff',\rho')=\rho'^{1/2}\eff'\rho'^{1/2}$ is strictly positive, a congruence of a strictly positive operator by an invertible one. Since $\rho'>0$ the map $s\mapsto\rho'^{s}$ is continuous, and so is $s\mapsto\Phi_s(\eff',\rho')$. Strict positivity is an open condition, so $\Phi_s(\eff',\rho')>0$ on an open interval about $\tfrac12$ and $\Phi_s(\eff,\rho)\ge0$ there.

Suppose instead that $\ker\eff$ is not $\rho$-invariant. It is then not $\log\rho$-invariant, so some unit $v\in\ker\eff$ has $(\log\rho)v\notin\ker\eff$, and Lemma~\ref{lem:rankdef} supplies $s\neq\tfrac12$ arbitrarily close to $\tfrac12$ at which $\Phi_s$ fails to be positive. No open interval about $\tfrac12$ can consist of admissible $s$.
\end{proof}

\begin{figure}[t]
\includegraphics[width=\columnwidth]{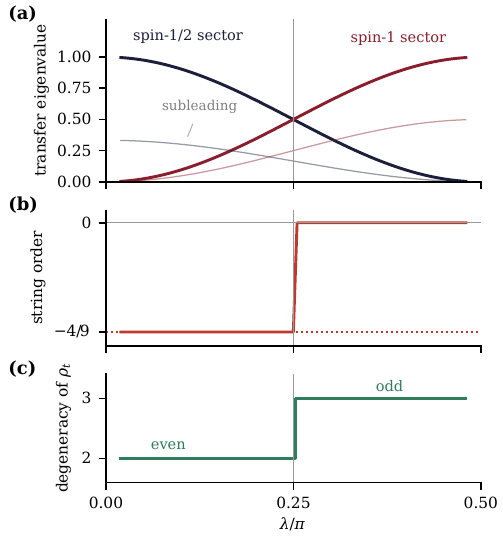}
\caption{A first-order transition of the record's class along Eq.~\eqref{eq:kraus}, which mixes a spin-$\tfrac12$ and a spin-$1$ block with weights $\cos\lambda$ and $\sin\lambda$. (a) Leading transfer eigenvalues of the two $SU(2)$ sectors, $\cos^{2}\lambda$ and $\sin^{2}\lambda$. Covariance forbids them to mix, so they cross without repulsion at $\lambda^{*}=\pi/4$, their separation $|\cos2\lambda|$ closing linearly. The subleading eigenvalue of each sector is in a lighter tone. Cross-sector eigenvalues, not drawn, lead only beyond $|\lambda-\lambda^{*}|\simeq0.03\pi$. (b) String order of the outcomes between sites $1$ and $60$, with the leading eigenvectors of $T$ as boundaries, so the step is not rounded \cite{supp}. It is $-4/9$ below $\lambda^{*}$ and zero above, the two-point function decaying on both sides. (c) Degeneracy of $\bar\rho$, two below $\lambda^{*}$ and three above, the parities Pollmann's criterion assigns to the classes. At $\lambda^{*}$ the channel has two fixed points, so $\bar\rho$ and the plotted points there are not unique.}
\label{fig:sptrecord}
\end{figure}

The condition is a compatibility between the post-selection and the past, since a post-selection may forbid whatever the filtered state does not mix with its complement and nothing else. A subspace inside one eigenspace of $\rho_t$, or spanned by eigenvectors from several, therefore leaves the interval open and narrowed only by the modular spread. A vector superposing two eigenvalues closes it \cite{supp}. The admissible forbidden subspaces are exactly the $\rho_t$-invariant ones. Under a symmetry the smallest that are symmetric as well are single multiplets, and it is their dimensions that the rest of the Letter is about.

\textit{What the collapse certifies.} A collapsed interval says more than that no ordering works. Let $v\in\ker\eff$. Since $\eff v=0$ and $\langle v|\eff=0$, the Margenau--Hill member $\Phi_0=\tfrac12(\eff\rho+\rho\eff)$ obeys $\langle v|\Phi_0|v\rangle=0$. A positive operator annihilates every vector on which its quadratic form vanishes, so $\Phi_0\ge0$ would force $\Phi_0 v=0$, hence $\eff\rho v=0$ and $\rho v\in\ker\eff$. Whenever Theorem~\ref{thm:invariant} denies the interval, $\Phi_0$ has a negative eigenvalue, witnessed by $w=v-t\Phi_0 v$ at small $t>0$. The hypothesis is $\rho\ge0$, so the certificate covers the low-rank state of a single record, unlike Theorem~\ref{thm:invariant}.

That eigenvalue is observable, since for unit $w$ and $\Pi=|w\rangle\langle w|$ the quantity $\langle w|\Phi_0|w\rangle=\mathrm{Re}\Tr[\eff\Pi\rho]$ is the real part of the weak value $\Pi_w=\Tr[\eff\Pi\rho]/\Tr[\eff\rho]$ up to the positive factor $\Tr[\eff\rho]$. The equivalence is due to Johansen and Luis \cite{johansen2004}. A negative real part lies outside the eigenvalue range $\{0,1\}$ of $\Pi$. Under the conditions of \cite{pusey2014,kunjwal2019} such an anomaly requires contextuality \cite{spekkens2008}, and how much follows from the same leak. Let $p_-$ be the joint probability that a weak pointer of width $\sigma$ reads negative and the post-selection succeeds, and $p_-^{\rm NC}$ the noncontextual bound \cite{kunjwal2019}. Optimizing over $\sigma$,
\begin{equation}
\max_\sigma\big(p_--p_-^{\rm NC}\big)\ \ge\ \frac{\|\eff\rho v\|^{4}}{32\pi(1-\Tr[\eff\rho])\,\|\Phi_0\|^{2}},
\label{eq:violation}
\end{equation}
with $v$ a forbidden direction and $\Phi_0$ the Margenau--Hill member (End Matter~\ref{app:contextual}). A qubit suffices. Prepare $\rho=\mathrm{diag}(0.97,0.03)$, forbid the Bloch direction at $0.31\pi$, and read a pointer $G_\sigma\propto e^{-x^{2}/2\sigma^{2}}$ of width $2$ \cite{kunjwal2019}. Then $\mathrm{Re}\,\Pi_w=-0.47$ and the bound is exceeded by $1.6\times10^{-2}$, which five standard deviations reach in $1.3\times10^{4}$ runs. Holding $\sigma$ and $\Pi$ fixed, noise $\eff\mapsto(1-2\epsilon)\eff+\epsilon I$ removes the violation at $\epsilon=0.11$ \cite{supp}. The collapse therefore certifies contextuality of the pair through a quantity the two-record configuration already yields.

\textit{The class of the record.} A system monitored by a channel $\mathcal{E}(X)=\sum_m A^{m}XA^{m\dagger}$ \cite{davies1976,ozawa1984,lindblad1976,breuer2002,belavkin1994,carmichael1993,wiseman2010} generates its record as a matrix product state in time whose tensors are the Kraus operators $A^{m}$, with the system as bond \cite{schon2005,osborne2010,verstraete2010,cirac2021}. Read as a hidden Markov model, a symmetry acts projectively on the bond and linearly on the outcomes \cite{souissi2026,souissi2026akl,souissi2026causal}. Averaged over records the filtered state becomes the fixed point $\bar\rho$ of $\mathcal{E}$, the left environment of that matrix product state and the leading eigenvector of the transfer matrix $T=\sum_m A^{m}\otimes\bar{A}^{m}$, whose gap sets the correlation time of the record. In the gauge $\sum_m A^{m\dagger}A^{m}=I$ the spectrum of $\bar\rho$ is the entanglement spectrum of the record at the cut. A single record instead gives a conditional state of low rank, so what follows concerns $\bar\rho$. The effect is supplied by the post-selection, an unconditioned future supplying $I$ in that gauge. When preparation, dynamics and post-selection respect a group $G$, both environments commute with its representation $V$ on the bond and Schur's lemma makes the eigenspaces of $\bar\rho$ sums of multiplets. The class of the record is the projective class of $V$, read from those degeneracies \cite{pollmann2010,chen2013}. For $G=SO(3)$ the class is $\mathbb{Z}_2$, half-integer against integer spin on the bond, and every multiplet of $\bar\rho$ has even dimension in the Haldane class and odd dimension in the trivial one. A forbidden subspace must then be invariant under $\bar\rho$ by Theorem~\ref{thm:invariant} and under $G$ by symmetry. It is a sum of multiplets, the minimal ones being single irreducible multiplets $W$ with
\begin{equation}
\dim W\equiv\begin{cases}0 \pmod 2 & \text{Haldane class,}\\ 1 \pmod 2 & \text{trivial class.}\end{cases}
\label{eq:parity}
\end{equation}
Since every admissible forbidden subspace is a sum of these, the Haldane class admits none of odd dimension while the trivial class admits some. That is the invariant in binary form.
The class does not appear in the values of the family, which at a renormalization fixed point are scalars by Schur's lemma, but in its structure. Even a $2\pi$ rotation of the outcome basis, under which the bond returns as $V(2\pi)=-1$, leaves every $\Phi_s$ invariant. The reason is that $\Phi_s$ is built from $V\cdot V^{\dagger}$. The multiplet structure of $\bar\rho$ is all the family sees.

\textit{Changing the class.} The class changes along a covariant family only where the channel has more than one fixed point, a non-injective matrix product state \cite{schuch2011,wolf2006}. In canonical form a matrix product tensor is a direct sum of injective blocks \cite{perezgarcia2007,wolf2006}, and injective channels of different class cannot meet at a unique fixed point \cite{chen2011,schuch2011}. The crossing takes this form of necessity. The simplest realization has a bond of spin-$\tfrac12\oplus$ spin-$1$ with spin-$1$ outcomes $m\in\{+1,0,-1\}$ and Kraus operators
\begin{equation}
A^{m}(\lambda)=\cos\lambda\,A^{m}_{1/2}\;\oplus\;\sin\lambda\,A^{m}_{1},
\label{eq:kraus}
\end{equation}
$A_{1/2}$ the AKLT tensors and $A_{1}$ the Clebsch--Gordan tensor coupling bond spin-$1$ to itself through the outcome \cite{supp}. Each block alone is trace preserving while the sum is not, so Eq.~\eqref{eq:kraus} is a tensor whose norm the sectors compete to carry. Rescaling by the dominant weight restores the gauge on the sector that survives. The AKLT state settled Haldane's conjecture that integer-spin chains have a gap \cite{haldane1983,aklt1987}, its sites being pairs of spin-$\tfrac12$ variables, one bound to each neighbour and one unpaired at the ends. Read along time, the same tensors give the AKLT chain as the record of a spin-$\tfrac12$ system emitting spin-$1$ outcomes, the unpaired half carrying the class. The Supplemental Material translates the two readings term by term \cite{supp}. Each block is covariant with $V=D^{(1/2)}\oplus D^{(1)}$, and the two lie in sectors of $V$ that covariance forbids to mix. Their leading eigenvalues $\cos^{2}\lambda$ and $\sin^{2}\lambda$ cross without repulsion \cite{vonneumann1929} at $\lambda^{*}=\pi/4$, separated by $|\cos2\lambda|$ (Fig.~\ref{fig:sptrecord}), which within $0.03\pi$ of $\lambda^{*}$ is the gap of $T$ itself. Below $\lambda^{*}$ the AKLT block dominates. The string order \cite{dennijs1989,kennedy1992} is $-4/9$, the two-point function decays, and $\bar\rho=I/2$ has even degeneracy. Above $\lambda^{*}$ the spin-$1$ block dominates, the string order vanishes, and $\bar\rho=I/3$ has odd degeneracy. At $\lambda^{*}$ the channel has two fixed points and the phases coexist, so the transition is first order, protected by the symmetry that keeps the sectors apart. Across it the parity in Eq.~\eqref{eq:parity} flips, and with it the admissible post-selections.

\textit{Away from the fixed points.} Deforming a block inside its phase gives $\bar\rho$ eigenspaces of dimensions $2$ and $4$. Over twenty random effects either multiplet gives half-width $0.50$ and no admissible choice falls below $0.06$, while a superposition closes it \cite{supp}. A deformed channel on the trivial side gives dimensions $3$, $1$, $5$ and the opposite parity \cite{supp}.

\textit{Two kinds of transition.} The transitions of $\Phi_s$ are Landau-type level crossings, with the sector holding $\lambda_{\min}$ as order parameter \cite{supp}. The record's transition has no local order parameter and lives in the transfer gap, and Eq.~\eqref{eq:parity} is where the two meet. Measurement-induced symmetry-protected phases \cite{lavasani2021,sang2021,morral2023,xiao2026} use the entanglement spectrum, which Theorem~\ref{thm:invariant} turns into a post-selection claim. Their object is the filtered state, positive by construction and lacking a future record. The class is read from the record, whose string operator is diagonal in the outcome basis, so its expectation is a classical average. Sampling gives $-0.4442(5)$ against $-4/9$ for the AKLT block and $0.0019(19)$ for the other \cite{supp}, with no post-selection problem \cite{morral2023}, the outcomes being classical data already. Positivity is one criterion among several and no member meets them all \cite{supp}, the physical state being fixed by the unraveling \cite{laverick2023}. Within positivity the symmetric split is singled out twice. No other member is positive throughout, and Proposition~\ref{prop:mult} makes it the only multiplicative one.

\textit{Conclusion.} A forbidden subspace leaves the smoothed state classical exactly when a full-rank filtered state does not mix it with its complement, and where it does the pair certifies contextuality. Under a symmetry the smallest subspaces that qualify are single multiplets, even-dimensional in the Haldane class and odd in the trivial one. Both are exact, and both concern the arrangement, whatever the run.

Symmetry-protected order is sought in a state and detected by an entanglement diagnostic. Here it came from a question with no topology in it. We asked how past and future records combine, and the combinations that stay classical know the topology of the record.

Three questions are left open. Whether some member of the ordering family can meet every reasonable condition at once \cite{supp} is not known. For groups with richer second cohomology than $SO(3)$ the minimal symmetric forbidden subspaces should carry more than a parity, and which dimensions occur is unknown. The transition here is first order because covariance keeps the sectors apart, and whether two classes can be joined continuously would say whether a record has phases in the full sense or only classes.

\begin{acknowledgments}
C.-F.~K. is grateful to the European Union and the Region R\'eunion, France (POE FEDER 2021--2027, n$^{\circ}$2025-0954-007180) for the funding support.
\end{acknowledgments}

\clearpage
\onecolumngrid
\begin{center}\textbf{End Matter}\end{center}
\twocolumngrid
\appendix

\section{The collapse certifies contextuality}
\label{app:contextual}

Both steps of the certificate are given here in full, with a bound on the anomaly they force.

\begin{theorem}[A collapsed interval forces a negative Margenau--Hill image]
\label{thm:mh}
Let $\rho\ge0$ and $0\le\eff\le I$ on a finite-dimensional space, with $\ker\eff\neq\{0\}$. If $\Phi_0(\eff,\rho)=\tfrac12(\eff\rho+\rho\eff)\ge0$ then $\ker\eff$ is invariant under $\rho$. Equivalently, a $\ker\eff$ that is not $\rho$-invariant gives $\Phi_0$ a negative eigenvalue. By Theorem~\ref{thm:invariant} that is exactly the case in which no open interval of orderings is admissible.
\end{theorem}

\begin{proof}
Johansen and Luis showed that an anomalous real part is equivalent to negativity of the Margenau--Hill distribution, for mixed preparations and imperfect detectors alike \cite{johansen2004}. Theorem~\ref{thm:mh} adds the source of that negativity and the bound below adds its size. Let $v\in\ker\eff$. Then $\eff v=0$, and since $\eff$ is Hermitian also $\langle v|\eff=(\eff v)^{\dagger}=0$. Hence
\begin{equation}
\langle v|\Phi_0|v\rangle=\tfrac12\langle v|\eff\rho+\rho\eff|v\rangle=0,
\label{eq:mhnull}
\end{equation}
each of the two terms vanishing because $\eff$ annihilates $v$ on one side or the other.

Suppose $\Phi_0\ge0$. Write $\Phi_0=B^{\dagger}B$ for some $B$, possible because $\Phi_0$ is Hermitian and positive. Then $\langle v|\Phi_0|v\rangle=\|Bv\|^{2}$, so Eq.~\eqref{eq:mhnull} gives $Bv=0$ and therefore $\Phi_0 v=B^{\dagger}Bv=0$. Expanding gives $\Phi_0 v=\tfrac12(\eff\rho v+\rho\eff v)=\tfrac12\eff\rho v$, so $\eff\rho v=0$ and $\rho v\in\ker\eff$. As $v$ was arbitrary in $\ker\eff$, the subspace is $\rho$-invariant. The contrapositive is the second statement.
\end{proof}

The proof also supplies the witness and its size. Suppose $\ker\eff$ is not $\rho$-invariant and fix $v\in\ker\eff$ of unit norm with $\eff\rho v\neq0$. Put $g=\|\Phi_0v\|=\tfrac12\|\eff\rho v\|$ and take the unnormalized vector $w_t=v-t\Phi_0v$ with $t>0$. Using Eq.~\eqref{eq:mhnull} and the Hermiticity of $\Phi_0$, the value is $-2tg^{2}+t^{2}\langle\Phi_0v|\Phi_0|\Phi_0v\rangle$, so
\begin{equation}
\langle w_t|\Phi_0|w_t\rangle\le g^{2}t\big(t\|\Phi_0\|-2\big),
\label{eq:witness}
\end{equation}
which is negative for every $0<t<2/\|\Phi_0\|$ and least at $t=1/\|\Phi_0\|$. Since $\|w_t\|^{2}=1+t^{2}g^{2}$ and $g\le\|\Phi_0\|$,
\begin{equation}
\lambda_{\min}(\Phi_0)\le-G,\qquad G\equiv\frac{\|\eff\rho v\|^{2}}{8\,\|\Phi_0\|},
\label{eq:mhbound}
\end{equation}
where $\|\Phi_0\|\le\|\eff\|\,\|\rho\|\le1$. The guarantee $G$ never exceeds half of what it guarantees. Since $\langle v|\Phi_0|v\rangle=0$, the variance inequality of Bhatia and Davis \cite{bhatiadavis2000} gives $\|\Phi_0v\|^{2}\le|\lambda_{\min}|\lambda_{\max}\le|\lambda_{\min}|\,\|\Phi_0\|$, so $G\le|\lambda_{\min}(\Phi_0)|/2$. The factor two is attained on a qubit with $\rho$ pure, $\eff=I-|v\rangle\langle v|$ and $|\langle\psi|v\rangle|^{2}=\tfrac34$, where the guarantee is $1/16$ and $\lambda_{\min}=-1/8$. That configuration also saturates the known ceiling $\mathrm{Re}\Tr[\eff\Pi\rho]\ge-1/8$ \cite{allahverdyan2014}, which runs the other way, so $G$ never exceeds $1/16$. Both statements were checked on $2\times10^{4}$ random instances \cite{supp}.

Here $\|\eff\rho v\|$ is the leak of Lemma~\ref{lem:rankdef} measured at the Margenau--Hill end, not near $s=\tfrac12$, and the difference matters. Near the symmetric split the expansion of $\cosh$ contributes a factor $\varepsilon^{2}$, so a tilt of angle $\delta$ out of invariance violates positivity only at order $\delta^{2}\varepsilon^{2}$. At $s=0$ the same tilt gives order $\delta^{2}$ with an $O(1)$ coefficient \cite{supp}.

The second step is the reading of $\Phi_0$. For a unit vector $w$ and $\Pi=|w\rangle\langle w|$,
\begin{equation}
\langle w|\Phi_0|w\rangle=\mathrm{Re}\Tr[\eff\Pi\rho],
\label{eq:mhweak}
\end{equation}
since $\Tr[\rho\eff\Pi]=\overline{\Tr[\eff\Pi\rho]}$ for Hermitian $\eff,\rho,\Pi$. Dividing by $\Tr[\eff\rho]>0$ makes this the real part of the weak value $\Pi_w=\Tr[\eff\Pi\rho]/\Tr[\eff\rho]$ with preparation $\rho$ and post-selection $\eff$. A negative eigenvalue of $\Phi_0$ therefore exhibits a projector whose weak value falls below the range $\{0,1\}$ of its eigenvalues, and Eq.~\eqref{eq:mhbound} says how far:
\begin{equation}
\mathrm{Re}\,\Pi_w\le-\frac{\|\eff\rho v\|^{2}}{8\,\Tr[\eff\rho]}.
\label{eq:wvbound}
\end{equation}
An anomalous weak value requires contextuality \cite{pusey2014}, in line with the general equivalence of negativity and contextuality \cite{spekkens2008}, but not unconditionally. Kunjwal, Lostaglio and Pusey give a noise-robust form \cite{kunjwal2019}. Couple $\Pi$ weakly to a pointer of width $\sigma$ and post-select on $\eff$. Let $p_-$ be the probability of a negative reading and a successful post-selection, and $p_F=\Tr[\eff\rho]$. Their first theorem assumes two operational facts. The pointer distribution is a mixture of shifted copies $q(x-1)$ and $q(x)$ of one distribution of median zero, keyed to a two-outcome measurement of $\Pi$, and the pointer-averaged channel is $(1-p_d)\mathcal{I}+p_d\mathcal{M}^{D}$. Dropping either admits a classical model of the anomaly. A noncontextual model then obeys
\begin{equation}
p_-\le p_F/2+(1-p_F)p_d.
\label{eq:klp}
\end{equation}
Neither assumption constrains the post-selection or requires a pure preparation, so $\eff$ and $\rho$ qualify as they stand. In the Gaussian realization the amplitude $G_\sigma(x)=(\pi\sigma^{2})^{-1/4}e^{-x^{2}/2\sigma^{2}}$ fixes both the pointer statistics and the disturbance. Its overlap $\int G_\sigma(x-1)G_\sigma(x)\,dx=e^{-1/4\sigma^{2}}$ gives $p_d=(1-e^{-1/4\sigma^{2}})/2$ and
\begin{equation}
p_-=\frac{p_F}{2}-\frac{1}{\sqrt{\pi}\,\sigma}\mathrm{Re}\Tr[\eff\Pi\rho]+o(1/\sigma),
\label{eq:klpquantum}
\end{equation}
so a negative real part violates Eq.~\eqref{eq:klp} once $\sigma$ is large enough. Maximizing over $\sigma$ gives a violation $2(\mathrm{Re}\Tr[\eff\Pi\rho])^{2}/\pi(1-p_F)$ to leading order, matched to three digits by direct evaluation in dimensions $2$ through $8$ \cite{supp}. With Eq.~\eqref{eq:mhbound} this is the bound of Eq.~\eqref{eq:violation}. A collapsed interval therefore forces not merely some violation of Eq.~\eqref{eq:klp} but one whose size follows from the arrangement alone. The same reference gives a version keyed to a nonzero imaginary part, which the collapse also supplies, since $[\eff,\rho]=0$ would leave $\ker\eff$ invariant.

The certificate is confined to the Margenau--Hill end, not by choice. Writing $\Tr[\eff\rho^{1-s}\Pi\rho^{s}]=\Tr[(\rho^{s}\eff\rho^{-s})\rho\Pi]$ exhibits $\Phi_s$ as the Margenau--Hill form of the flowed effect $\rho^{s}\eff\rho^{-s}$, which is not Hermitian away from $s=0$ and so names no weak value. For one Wishart pair in dimension four the operator norm of its antihermitian part grows from zero at $s=0$ to $3.3$ at $s=\tfrac12$, and its Hermitian part loses positivity by $s=0.3$ \cite{supp}. Theorem~\ref{thm:mh} makes this harmless. It carries the collapse from the symmetric split, where Theorem~\ref{thm:invariant} applies, to the Margenau--Hill end where the weak value is defined.

A commuting pair leaves $\ker\eff$ invariant, so Eq.~\eqref{eq:mhbound} vanishes and forces no anomaly. The converse of Theorem~\ref{thm:mh} is false. For noncommuting $\eff$ and $\rho$ the operator $\Phi_0$ is generically negative even on an invariant kernel, in a fraction of random instances that grows with dimension \cite{supp}, so the certificate runs one way only.

\section{Multiplicativity selects the symmetric split}
\label{app:mult}

A criterion that selects a point is worth little if the point moves when the criterion changes. Two systems monitored separately ought to have a smoothed description that is the product of the two, as the classical smoother does without comment. Asking the same of Eq.~\eqref{eq:phis} is asking which member respects the tensor product.

\begin{proposition}
\label{prop:mult}
Let $\rho_i>0$ and $0\le\eff_i\le I$ on finite-dimensional spaces of dimension at least two. Then $\Phi_s(\eff_1\otimes\eff_2,\rho_1\otimes\rho_2)=\Phi_s(\eff_1,\rho_1)\otimes\Phi_s(\eff_2,\rho_2)$ for all such states and effects if and only if $s=\tfrac12$.
\end{proposition}

\begin{proof}
Write $A=\rho_1^{s}\eff_1\rho_1^{1-s}$, $B=\rho_2^{s}\eff_2\rho_2^{1-s}$, and let $\Delta_s$ denote the difference between the two sides of the claimed identity. Since $(\rho_1\otimes\rho_2)^{s}=\rho_1^{s}\otimes\rho_2^{s}$, the product system gives $\tfrac12(A\otimes B+A^{\dagger}\otimes B^{\dagger})$ while the factors give $\tfrac14(A+A^{\dagger})\otimes(B+B^{\dagger})$. Hence
\begin{equation}
\Delta_s=\tfrac14\,(A-A^{\dagger})\otimes(B-B^{\dagger}).
\label{eq:multgap}
\end{equation}
A tensor product vanishes only if a factor does, and the factors range independently, so $\Delta_s$ vanishes throughout exactly when $\rho^{s}\eff\rho^{1-s}$ is Hermitian for every state and effect on one system. In the eigenbasis $\rho=\sum_j p_j|j\rangle\langle j|$ that operator has entries $p_j^{s}\eff_{jk}p_k^{1-s}$ while its adjoint has $p_k^{s}p_j^{1-s}\eff_{jk}$, using $\eff_{kj}=\overline{\eff_{jk}}$. Hermiticity therefore needs $(p_j/p_k)^{2s-1}=1$ for every pair with $\eff_{jk}\neq0$. Dimension at least two allows $p_1\neq p_2$ and $\eff_{12}\neq0$, forcing $s=\tfrac12$. Conversely at $s=\tfrac12$ the operator $\rho^{1/2}\eff\rho^{1/2}$ is Hermitian, $A=A^{\dagger}$, and $\Delta_s$ vanishes identically.
\end{proof}

The same computation says when the other members fail. For $s\neq\tfrac12$ the condition $(p_j/p_k)^{2s-1}=1$ forces $[\eff,\rho]=0$, so by Eq.~\eqref{eq:multgap} an ordering away from the symmetric split is multiplicative on a given pair of systems unless both pairs fail to commute. The two selections of $s=\tfrac12$ are independent: positivity concerns the spectrum of a single image and says nothing about tensor products, while Eq.~\eqref{eq:multgap} holds whether or not either side is positive. A third selection comes from outside. The member $\Phi_{1/2}=\sqrt{\rho}\,\eff\sqrt{\rho}$ is what Bayesian retrodiction gives through the Petz map \cite{liu2025}, so all three agree. None makes $s=\tfrac12$ the physical smoothed state, which the unraveling fixes \cite{laverick2023}. At the Margenau--Hill end it is Hermitian but not positive \cite{laverick2025}. The extended discussion is in the Supplemental Material \cite{supp}.

\section{Methods}
\label{app:methods}

Parameter values are in the Supplemental Material \cite{supp}.

Boundaries in $s$ come from bisection on $\lambda_{\min}[\cosh(\varepsilon\mathcal{K})(\eff)]$, whose sign decides positivity by Eq.~\eqref{eq:reduction}. The crossing is transversal for generic $(\eff,\rho)$ and the negativity vanishes linearly, with fitted exponents $1.004$ to $1.007$, the generic behaviour of a simple zero and not a critical exponent. The width is governed by the smallest eigenvalue of $\eff$ against the spread of $\log\rho$ and not by the dimension. For $\rho$ maximally mixed $\mathcal{K}=0$ and every $s$ is admissible at every $d$, and the shrinkage seen in Wishart ensembles is the conditioning of the ensemble, a fixed-spectrum ensemble giving full width throughout.

Covariant channels are built from Clebsch--Gordan tensors on a bond $\bigoplus_j V_j$, which makes $\sum_{m'}D^{(1)}(g)_{mm'}A^{m'}=V(g)A^{m}V(g)^{\dagger}$ exact to $10^{-16}$ by construction. The environments $\bar\rho$ and $\bar{\eff}$ are the leading eigenvectors of $T=\sum_m A^{m}\otimes\bar{A}^{m}$ and of $T^{\dagger}$, reshaped and Hermitized, and the gap of $T$ is the difference of the two largest eigenvalue moduli. Minus the logarithm of their ratio gives the inverse correlation length, a different quantity here because the leading eigenvalue is $\cos^{2}\lambda$ or $\sin^{2}\lambda$ and not one. String order and two-point functions are evaluated exactly by inserting the operators into $T$ at $L=60$. Sampled values come from Born sampling, $10^{6}$ records of length $40$ for the AKLT channel and $1.2\times10^{5}$ for the spin-$1$ channel, each block normalized to be trace preserving beforehand and the string taken between sites $8$ and $32$. The AKLT value then sits half a standard deviation from $-4/9$. A forbidden subspace is imposed by projecting a random positive operator onto its orthogonal complement, which makes it the null space of the resulting effect.

\end{document}